\documentclass{IEEEtran4PSCC}

\ifCLASSINFOpdf
   \usepackage[pdftex]{graphicx}
\else
   \usepackage[dvips]{graphicx}
\fi

\usepackage[cmex10]{amsmath}
\usepackage{amsthm}
\usepackage{algorithm,algorithmic}
\newtheorem{theorem}{Theorem}

\newtheorem{corollary}{Corollary}
\usepackage{bm}
\usepackage{bbm}
\usepackage{xcolor}
\graphicspath{{figure/}}
\usepackage{float}
\usepackage[caption=false,font=footnotesize]{subfig}
\usepackage{amssymb,amsfonts}
\usepackage{cite}

\makeatletter
\let\old@ps@headings\ps@headings
\let\old@ps@IEEEtitlepagestyle\ps@IEEEtitlepagestyle
\def\psccfooter#1{%
    \def\ps@headings{%
        \old@ps@headings%
        \def\@oddfoot{\strut\hfill#1\hfill\strut}%
        \def\@evenfoot{\strut\hfill#1\hfill\strut}%
    }%
    \def\ps@IEEEtitlepagestyle{%
        \old@ps@IEEEtitlepagestyle%
        \def\@oddfoot{\strut\hfill#1\hfill\strut}%
        \def\@evenfoot{\strut\hfill#1\hfill\strut}%
    }%
    \ps@headings%
}
\makeatother

\psccfooter{%
        \parbox{\textwidth}{\hrulefill \\ \small{22nd Power Systems Computation Conference} \hfill \begin{minipage}{0.2\textwidth}\centering \vspace*{4pt} \includegraphics[scale=0.06]{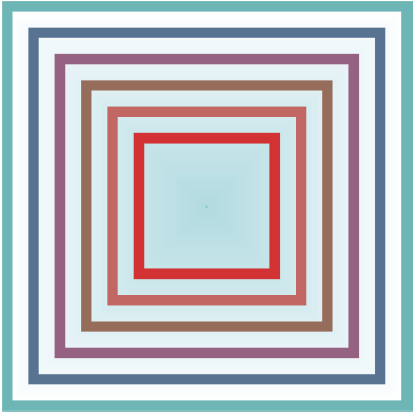}\\\small{PSCC 2022} \end{minipage} \hfill \small{Porto, Portugal --- June 27 -- July 1, 2022}}%
}

\begin{document}
\title{Decentralized Safe Reinforcement Learning for Voltage Control}

\author{
\IEEEauthorblockN{Wenqi Cui, Jiayi Li and Baosen Zhang}
\IEEEauthorblockA{Electrical and Computer Engineering,
University of Washington, 
Seattle, WA \\
\{wenqicui, ljy9712, zhangbao\}@uw.edu}
}

\maketitle

\begin{abstract}
 Inverter-based distributed energy resources provide the possibility for fast time-scale voltage control by quickly adjusting their reactive power. The power-electronic interfaces allow these resources to realize almost arbitrary control law, but designing these decentralized controllers is nontrivial. Reinforcement learning (RL) approaches are becoming increasingly popular to search for policy parameterized by neural networks.
 It is difficult, however, to enforce that the learned controllers are safe, in the sense that they may introduce instabilities into the system. 

This paper proposes a safe learning approach for voltage control.
We  prove  that  the  system  is  guaranteed to  be  exponentially  stable  if  each  controller  satisfies  certain Lipschitz constraints. 
The set of Lipschitz bound is optimized to enlarge the search space for neural network controllers.
We explicitly engineer the structure of neural network controllers such that they satisfy the Lipschitz constraints by design. A decentralized RL framework is constructed to train local neural network controller at each bus in a model-free setting.





\end{abstract}

\begin{IEEEkeywords}
Safe RL, Voltage control, Stability, Decentralized Learning
\end{IEEEkeywords}

\thanksto{\noindent The authors are partially supported by NSF grants ECCS-1942326, CNS-1931718 and the Washington Clean Energy Institute.}

\section{Introduction} \label{sec:intro}
Distributed energy resources (DERs) such as rooftop solar PV, electric vehicles and battery storage are growing at an increasing pace. For example, solar capacity had almost 50\% yearly growth in 2021~\cite{Solar21}, which is by far the fastest among all renewable resources. Most of these growth are occurring in the distribution network, the low voltage network that connects customers to substations. 

High variability of solar PV and sudden change in load due to electric vehicles and storage can lead to large voltage fluctuations. These fluctuations occur at timescales much faster than the conventional mechanical control devices such as tap-changing transformers. Instead,  power electronic devices allow flexible and frequent control actions without degrading lifetime. Consequently, there have been growing interests to use the power electronic inverters on the DERs themselves to provide voltage control~\cite{turitsyn2011options,yeh2012adaptive,zhang2014optimal,li2014real,bolognani2013distributed}.

Since most distribution networks are not yet equipped with real-time communication infrastructure, voltage control strategies should use local measurements available at each bus. 
More specifically, controllers need to operate at an iterative fashion~\cite{zhu2015fast,qu2019optimal}, successively updating their control actions based on each measurement. Designing such decentralized controller is a nontrivial problem. Linear controllers can be far from optimal, even for quadratic costs. 
Therefore, neural networks have been used to parametrize the controllers to fully utilize the capabilities of the inverters~\cite{hsu1998combined,toma2008decentralized,shen2019distributed}.

\begin{figure}[ht]	
	\centering
	\includegraphics[width=3.2in]{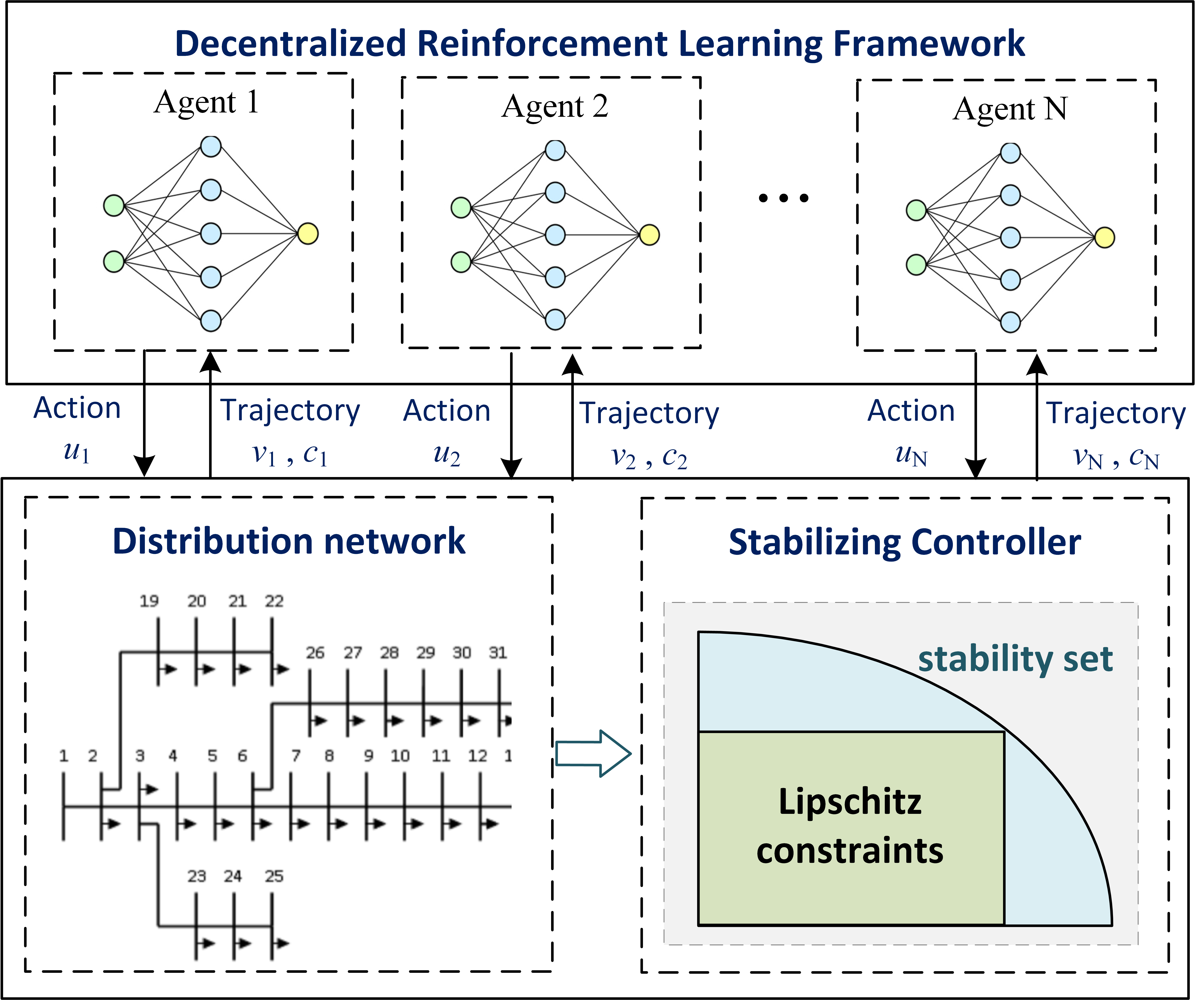}
	\caption{Proposed decentralized safe RL approach for optimal voltage control. We prove that the system is guaranteed to be exponentially
stable if each controller satisfies certain Lipschitz constraints.
The neural network controllers are engineered to satisfy these Lipschitz constraints by design, and is updated from local trajectories with a decentralized RL framework. }
	\label{fig:Structure_safeRL}
\end{figure}

Reinforcement learning algorithms are proposed to train the neural network controllers with trajectory measurements. This provides the advantages of updating neural networks in a model-free setting, i.e., eliminating the requirement on system parameters and communications~\cite{sutton2018reinforcement}. Many algorithms, such as deep Q learning~\cite{yang2019two}, actor-critic~\cite{gao2021consensus}, DDPG~\cite{duan2019deep}, have been applied to the control of tap-changing transformers or inverter based resources. Since the control actions are taken in an iterative fashion, it creates a dynamical system, whose transition depends on the actions and the underlying physical distribution network.
The key constraint on the controllers is that they do not destabilize the system. However, most works neglect the stability requirement and currently this stability condition is checked through simulations~\cite{ yang2019two,gao2021consensus}. Considering that voltage control is implemented locally without real-time communication, formal guarantees on stability are required in practice.

This paper presents a decentralized safe learning method, which guarantees the learned neural network would maintain the stability of iterative voltage control dynamics. We prove that the system is guaranteed to be exponentially stable if each controller satisfies certain Lipschitz constraints. We optimize the set of Lipschitz bounds to enlarge the search space of controllers. 
On this basis, we propose to engineer the structure of neural network controllers such that they can satisfy the Lipschitz constraints by design.  A decentralized RL framework is constructed to train neural network controller locally at each bus with policy gradient algorithm.  The structure of the proposed approach is illustrated in Fig.~\ref{fig:Structure_safeRL}.

Case studies show that the controllers learned with stability constraints outperform those with linear controllers and unconstrained neural network controllers. Interestingly, we also observe good learning convergence of the controllers in a model-free setting, even though they interact through the underlying distribution network. Code and data are available at https://github.com/Safe-RL-Power-Systems-Control/Voltage-Control.

The paper is organized as follows. Section~\ref{sec:model} introduces the  model and the optimal voltage control problem. Section~\ref{sec:stability} gives the main theorems governing the structure of a stabilizing controller and derives the optimal Lipschitz bounds. Section~\ref{sec:training} illustrates the decentralized safe RL framework for training a stabilizing neural network controller locally at each bus.  Section~\ref{sec:results} shows the simulation results and Section~\ref{sec:conclusions} concludes the paper.

\section{Model} \label{sec:model}
A standard requirement for distribution network is that voltages should deviate no more than 5\% from their rated values at all buses~\cite{national1996american}. 
For example, if the rated voltage is 110 V, then the actual voltages should be in the interval from 104.5 V to 115.5 V.
For simplicity, we normalize the units such that the reference value for voltage is 1 p.u. 
For a power network with N buses, let $\bm{v}$ be the voltage vector where $v_i$ is the voltage at bus $i$. Let $\bm{p}$ be active power and $\bm{q}$ be reactive power. The voltage of the system follows the LinDistFlow model: 
\begin{equation}\label{eq:Dyn_voltage}
\bm{v}=\mathbf{R} \bm{p}+\mathbf{X} \bm{q}+\mathbbm{1}
\end{equation}
where $\mathbbm{1}$ is the all one's vector and $\mathbf{R}$ and $\mathbf{X}$ are positive definite matrices describing the network~\cite{zhu2015fast}. 
The active power depends on external environment and is uncertain and variable. 
The reactive power comes from phase offsets and is controllable, subject to some actuation constraints~\cite{turitsyn2011options}.

This work focuses on optimizing the control of $\bm{q}$ through inverter-based resources. The aim of voltage regulation is to control $\bm{q}$ such that $\bm{v}$ is close to its reference value. 
Due to the lack of communication in many distribution systems, $\bm{q}$ needs to be successively updated based on the local voltage measurements. 
Denote $u_{i}(v_i)$ as the control law for each bus $i=1,\cdots,N$, which is a mapping from the voltage to reactive power. 
Let $v_{i, t}$ be the local voltage at the bus $i$ at the $t$-th iteration step, and denote   $\bm{u}_t=(u_{1}(v_{1, t}), \cdots, u_{N}(v_{N, t}))$. 
We update $\bm{q}$ and $\bm{v}$ iteratively as
\begin{subequations}\label{eq:Dynamic}
\begin{align}
\bm{q}_{ t+1}&=\bm{q}_{t}-\bm{u}_t , \label{subeq:Dynamic_q}\\
 \bm{v}_{t+1}&=\mathbf{R} \bm{p}+\mathbf{X} \left(\bm{q}_{t}-\bm{u}_t\right) + \mathbbm{1}, \label{subeq:Dynamic_v}
\end{align}
\end{subequations}



\subsection{Optimal voltage control}

Our objective is to optimize the $\bm{u}_t$ to minimize cost in $\bm{v}$ and $\bm{q}$ defined as $C(\bm{u})$,
subject to the iterative update rule and the saturation limit on $\bm{u}_t$.  The optimization problem is 
\begin{subequations}\label{eq:Optimization}
\begin{align}
\min_{\mathbf{u}} & \quad C(\bm{u}) \label{subeq:Optimization_obj}\\
\mbox{s.t. } & \bm{q}_{ t+1}=\bm{q}_{t}-\bm{u}_t \label{subeq:Optimization_q}\\
 &\bm{v}_{t+1}=\mathbf{R} \bm{p}+\mathbf{X} \left(\bm{q}_{t}-\bm{u}_t\right) + \mathbbm{1} \label{subeq:Optimization_V}\\
& \underline{\bm{u}} \leq \bm{u}_t\leq \overline{\bm{u}}\label{subeq:Optimization_bound}\\
& \bm{u}_t \text{ is stabilizing}\label{subeq:Optimization_stability}
\end{align}
\end{subequations}
where constraints~\eqref{subeq:Optimization_q}-\eqref{subeq:Optimization_stability} hold for the iteration step $t$ from $0$ to $T$. The cost typically trades off between driving voltage to the reference value and the control effort. The deviation of voltage can typically be quantified as two-norm, one-norm or infinity-norm of the sequence of $\bm{v}_t$~\cite{zhang2014optimal, vaccaro2011decentralized, jafari2018optimal}.
The control effort depends on the type of resources and can be both quadratic~\cite{zhao2014design, mallada2017optimal} and non quadratic ones~\cite{shi2017using, vaccaro2011decentralized, jafari2018optimal}. For example, control effort from batteries is commonly defined as one-norm of actions since charging/discharging power affects cycle-depth linearly~\cite{shi2017using, jafari2018optimal}. The proposed safe RL approach works for all types of  cost functions listed above.
The lower and upper bound for the control action at bus $i$ are $\underline{u}_i$ and $\overline{u}_i$, respectively.

The controllers $\bm{u}$ are conventionally designed to be linear (up to a thresholding by~\eqref{subeq:Optimization_bound}), which does not leverage the capability of inverter-based resources in implementing almost arbitrary control laws~\cite{johnson2013synchronization}. 
To design a flexible non-linear control law for inverter-based resources, we parameterize each controller $u_i(v_i)$ as a neural network with weight $\theta_i$, sometimes written as $u_{\theta_i}(v_i)$.

However, there remain two challenges. First, due to the lack of communication, neural network controller needs to be trained decentralizedly in each bus with local observations of voltages. Second, even if the controller is optimized and implemented locally, they need to stabilize the system defined by \eqref{subeq:Optimization_q} and \eqref{subeq:Optimization_V}.
In the next sections, we show how to design the local neural network controllers that guarantee the stability of this system, and how to train the controllers through decentralized reinforcement learning.

In this paper, we assume that the topology and parameter information of the distribution system is available. That is, we know $\mathbf{X}$, but there is no real-time communication between the buses. This assumption comes from the fact that $\mathbf{X}$ (and $\mathbf{R}$) can be estimated using smart meter data collected over a period of time~\cite{yu2018patopaem,zhang2020topology,lin2021data}, where the communication rate can be quite slow (e.g., once per day~\cite{kavetcharacterization}). Therefore, design of the controllers $\bm{u}$ can depend on $\mathbf{X}$, but the dependence must be determined offline.

\section{Stabilizing controller} \label{sec:stability}


In this section, we derive the properties of a stabilizing local controller from the Lyapunov stability theory and standard nonlinear system theory. 
We engineer the structure of neural network to satisfy these structure properties and thus guarantee the stability of the system.  

\subsection{Reduced-order system}
We can simplify the dynamics in \eqref{eq:Dynamic} by shifting the origin of the system. 
Denote $\Delta \bm{v}_t $ as the difference between voltage and its reference value $\mathbbm{1}$ at time $t$. Then we have 
\begin{equation}\label{eq:reduced_transition}
\begin{split}
\Delta\bm{v}_t &= \bm{v}_t-\mathbbm{1} \\
&= \bm{R}\bm{p} + \mathbf{X} \bm{q}_{t} \\
&=  \bm{R}\bm{p} + \mathbf{X}( \bm{q}_{t-1} - \bm{u}_{t})\\
&= (\bm{R}\bm{p} + \mathbf{X} \bm{q}_{t-1}) - \mathbf{X}\bm{u}_{t} \\
&=  \Delta\bm{v}_{t-1} - \mathbf{X}\bm{u}_{t} 
\end{split}
\end{equation}

Instead of the iteration with both $\bm{q}_t$ and $\bm{v}_t$ being the state variables, it suffices to study the dynamics in $\Delta \bm{v}_t$. In the next sections, we drop $\Delta$ and use $\bm{v}_t$ to denote $\Delta \bm{v}_t$.

\subsection{Structure property of a stabilizing controller}

The structure property of a stabilizing controller is obtained from Theorem~\ref{theorem:Exponential_stable}. It shows that as long as each controller $u_i$ satisfies the Lipschitz constraints, the system is guaranteed to be locally exponentially stable.

\begin{theorem}\label{theorem:Exponential_stable}
Suppose a vector $ \bm{k}=(k_1,\cdots, k_N)$ satisfies  $0 \prec \text{diag}(\bm{k}) \prec 2\mathbf{X}^{-1}$. Then if the derivative of controller satisfies $u_{i}(0)=0$ and $0<\frac{\mathrm{d}u_{i}(v_{i}) }{\mathrm{d} v_{i}}<k_i$ for all $i=1,\cdots,N$, the equilibrium point $\bm{v}=0$ of the dynamic system in  \eqref{eq:reduced_transition} is locally exponentially stable.
\end{theorem}

\begin{proof}
The Jacobian of the state transition dynamics in~\eqref{eq:reduced_transition} is 
\begin{equation}\label{eq:jacobian}
    \mathbf{J}(\bm{v})=\mathbf{I}-\mathbf{X}\nabla_{\bm{v}}\bm{u}
\end{equation}
where $\nabla_{\bm{v}}\bm{u}$ is the gradient of control action $\bm{u}$ with respect to $\bm{v}$ defined as
\begin{equation}\label{eq:J_v_w}
\nabla_{\bm{v}}\bm{u}=
\begin{bmatrix}
\begin{smallmatrix}
    \frac{\mathrm{d}u_1(v_1) }{\mathrm{d} v_1} & & \\
    & \ddots & \\
    & & \frac{\mathrm{d}u_{N}(v_{N}) }{\mathrm{d} v_{N}} 
\end{smallmatrix}
\end{bmatrix}.
\end{equation}

To guarantee an exponentially stable system around the equilibrium, the goal is to show that all the eigenvalues of $\mathbf{J}(\bm{v})$ have magnitude less than 1. To this end, we first show that  the eigenvalues of $\mathbf{J}(\bm{v})$ are the same as that of $\mathbf{I}-(\nabla_{\bm{v}}\bm{u})^{\frac{1}{2}}\mathbf{X}(\nabla_{\bm{v}}\bm{u})^{\frac{1}{2}}$.

Let $(\lambda, w)$ be an eigenpair for $\mathbf{I}-\mathbf{X}(\nabla_{\bm{v}}\bm{u})$. That is, $(\mathbf{I}-\mathbf{X}(\nabla_{\bm{v}}\bm{u}))w=\lambda w$. Then, we have 
\begin{equation}
\begin{split}
    &(\mathbf{I}-(\nabla_{\bm{v}}\bm{u})^{\frac{1}{2}}\mathbf{X}(\nabla_{\bm{v}}\bm{u})^{\frac{1}{2}})(\nabla_{\bm{v}}\bm{u})^{\frac{1}{2}}w\\
    =& (\nabla_{\bm{v}}\bm{u})^{\frac{1}{2}}w-(\nabla_{\bm{v}}\bm{u})^{\frac{1}{2}}\mathbf{X}(\nabla_{\bm{v}}\bm{u})w\\
    =&(\nabla_{\bm{v}}\bm{u})^{\frac{1}{2}}(\mathbf{I}-\mathbf{X}(\nabla_{\bm{v}}\bm{u}))w\\
    =& \lambda (\nabla_{\bm{v}}\bm{u})^{\frac{1}{2}} w
\end{split}
\end{equation}
Therefore, $(\lambda, (\nabla_{\bm{v}}\bm{u})^{\frac{1}{2}}w)$ is an eigenpair for $\mathbf{I}-(\nabla_{\bm{v}}\bm{u})^{\frac{1}{2}}\mathbf{X}(\nabla_{\bm{v}}\bm{u})^{\frac{1}{2}}$. To prove that the eigenvalue of $\mathbf{J}(\bm{v})$ to be strictly smaller than 1, it suffices to show that $-\mathbf{I} \prec \mathbf{I}-(\nabla_{\bm{v}}\bm{u})^{\frac{1}{2}}\mathbf{X}(\nabla_{\bm{v}}\bm{u})^{\frac{1}{2}} \prec \mathbf{I}$.

By picking the controller $\bm{u}$ such that $0<\frac{\mathrm{d}u_{i}(v_{i}) }{\mathrm{d} v_{i}}<k_i$ for all $i=1,\cdots,N$ and $0 \prec \text{diag}(\bm{k}) \prec 2\mathbf{X}^{-1}$, we have $0 \prec \nabla_{\bm{v}}\bm{u} \prec 2\mathbf{X}^{-1}$ and thus $(\nabla_{\bm{v}}\bm{u})^{-1} \succ  \frac{1}{2}\mathbf{X}$. Since  $\nabla_{\bm{v}}\bm{u}\succ 0$ is diagonal, we then have $(\nabla_{\bm{v}}\bm{u})^{\frac{1}{2}}\mathbf{X}(\nabla_{\bm{v}}\bm{u})^{\frac{1}{2}} \prec 2\mathbf{I}$ and thus $-\mathbf{I}\prec \mathbf{I}-(\nabla_{\bm{v}}\bm{u})^{\frac{1}{2}}\mathbf{X}(\nabla_{\bm{v}}\bm{u})^{\frac{1}{2}}\prec \mathbf{I}$. The right side inequality holds because $\mathbf{X}\succ 0$.

\end{proof}


\subsection{Optimizing search space for neural network controllers}
Note that all the feasible stabilizing $\bm{u}$ are in a convex set described by  $\mathcal{S} = \left \{\nabla_{\bm{v}}\bm{u}| 0 \prec  \nabla_{\bm{v}}\bm{u} \prec 2\mathbf{X}^{-1} \right \} $. Since there is no communication between buses during training, each $\frac{\mathrm{d}u_{i}(v_{i}) }{\mathrm{d} v_{i}}$ needs to be bounded by a separate $k_i$ for bus $i=1, \cdots, N$. Therefore, the search space for neural network controllers is constrained by the selection of $\bm{k}$. A uniform bound  $k_i=\frac{2}{\lambda_{max}(\mathbf{X})}$ can be found in literatures~\cite{zhu2015fast}, but it might be too conservative since $\mathcal{S}$ may be much larger than the region described by $\mathcal{D} = \left \{\nabla_{\bm{v}}\bm{u}|0 \prec \nabla_{\bm{v}}\bm{u}\prec \frac{2}{\lambda_{max}(\mathbf{X})}\mathbf{1} \right \}$.

Here we show an illustration on a three-bus system (with the first bus as the feeder) where $\mathbf{X}=\begin{bmatrix}
\begin{smallmatrix}
0.20 & -0.16\\
-0.16 & 0.97
\end{smallmatrix}
\end{bmatrix} $. For different Lipschitz bounds on controllers,  feasible regions  for $\nabla_{\bm{v}}\bm{u}$ are shown in Fig.~\ref{fig:search_region}.

\begin{figure}[ht]	
	\centering
	\includegraphics[width=3.5in]{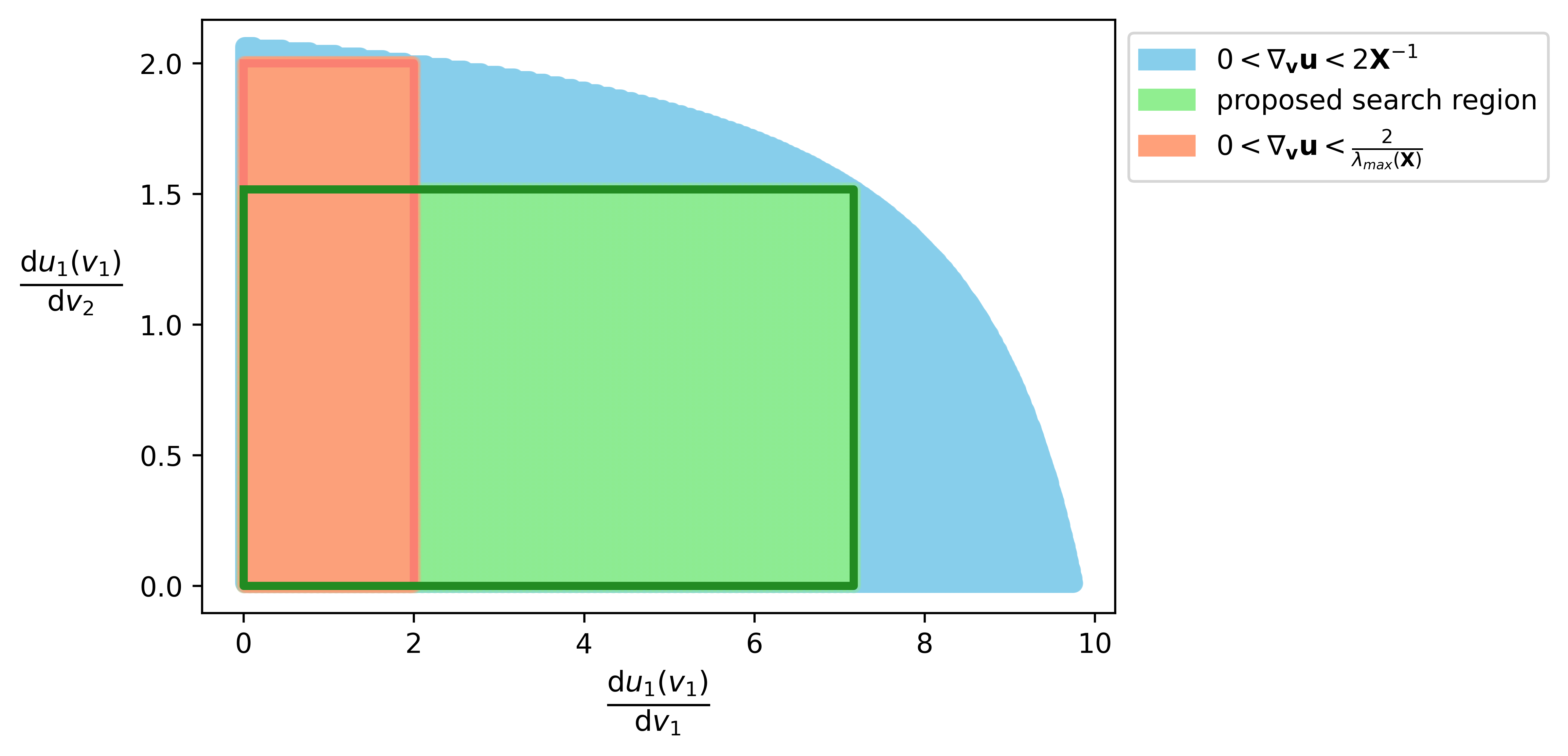}
	\caption{Feasible search space comparisons for controllers. The blue area is the set of all feasible $\bm{u}$ in $\mathcal{S} = \left \{\nabla_{\bm{v}}\bm{u}| 0 \prec  \nabla_{\bm{v}}\bm{u} \prec 2\mathbf{X}^{-1} \right \} $.
    The orange area is the search space with uniform Lipschitz bounds defined as $\mathcal{D} = \left \{\nabla_{\bm{v}}\bm{u}|0 \prec \nabla_{\bm{v}}\bm{u}\prec \frac{2}{\lambda_{max}(\mathbf{X})}\mathbf{1} \right \}$, which is the largest square within blue region but is only a very small subset of  $\mathcal{S}$. 
    With each controller being trained independently, 
    it is natural to consider some larger non-uniform search space such as the green area.}
	\label{fig:search_region}
\end{figure}

The blue area demonstrates the space of controllers constrained by $\mathcal{S} = \left \{\nabla_{\bm{v}}\bm{u}| 0 \prec  \nabla_{\bm{v}}\bm{u} \prec 2\mathbf{X}^{-1} \right \} $.
The orange area is the space defined by $\mathcal{D} = \left \{\nabla_{\bm{v}}\bm{u}|0 \prec \nabla_{\bm{v}}\bm{u}\prec \frac{2}{\lambda_{max}(\mathbf{X})}\mathbf{1} \right \}$, which is the largest square within blue region but is only a very small subset of blue area for $\mathcal{S}$. Note that the axes are scaled so the orange one does not look like a square.
With each controller being trained independently, 
it is natural to consider some larger non-uniform search space such as the green area by choosing different $\bm{k}$. We may choose a $\bm{k}^*$ such that the search space $\left \{\nabla_{\bm{v}}\bm{u}|0 \prec \nabla_{\bm{v}}\bm{u}\prec \bm{k}^* \right \}$ is the largest rectangular volume inside blue space, denoted as $\prod_{i=1}^{N}k_i$. 

The volume is not a convex function in $\bm{k}$, but we can apply a simple log trick and solve the following optimization problem:
\begin{subequations}\label{eq:OPT_Bound}
\begin{align}
\max_{\mathbf{k}} & \sum_{i=1}^{N} w_i \log(k_i) \label{subeq:OPT_Bound_obj}\\
\mbox{s.t. } & 0 \prec 
\begin{bmatrix}
\begin{smallmatrix}
    k_1 & & \\
    & \ddots & \\
    & & k_N\\
\end{smallmatrix}
\end{bmatrix}
\prec 2\mathbf{X}^{-1}
\label{subeq:OPT_Bound_stability}
\end{align}
\end{subequations}
where $w_1, \cdots, w_N$ are the coefficients to represent the relative importance of buses. For example, if bus $j$ has none or very limited capacity for voltage regulation, $w_j$ is set to be small. 
If bus $j$ is the source node of a branch, $w_j$ can be set to be larger to speed up the convergence of voltage at the source node and thus help the convergence of following branches. 
In reality, this set of coefficients can be adjusted according to the optimization results. 
For the controller $u_i$ that the derivative $\frac{\mathrm{d}u_{i}(v_{i}) }{\mathrm{d} v_{i}}$ is far from being bounded by $k_i$, the coefficient $w_i$ can be adjusted to be smaller to give larger searching freedom to other buses.

\subsection{Design of stabilizing neural network controllers }
From Theorems~\ref{theorem:Exponential_stable}, the structural property of locally exponentially stabilizing controllers is derived in Corollary~\ref{corollary: stablizing_structure}. We aim to engineer the neural networks to satisfy these structural property in Corollary~\ref{corollary: stablizing_structure} by design.

\begin{corollary}\label{corollary: stablizing_structure}
The condition for a locally exponentially stabilizing controller in Theorems~\ref{theorem:Exponential_stable} is equivalent to: 
\begin{enumerate}
    \item $u_{\theta_i}(v_i)$ has the same sign as $v_i$
    \item $u_{\theta_i}(v_i)$ is monotonically increasing
    \item $\frac{\mathrm{d}u_{\theta_i}(v_{i}) }{\mathrm{d} v_{i}}<k_i$.
\end{enumerate}
\end{corollary}

The first two requirements are equivalent to designing a monotonically increasing function through the origin. This is constructed by decomposing the function into a positive and a negative part as $f_i(v_i)=f_i^+(v_i)+f_i^-(v_i)$, where $f_i^+(v_i)$ is monotonically increasing for $v_i>0$ and zero when $v_i\leq0$; $f_i^-(v_i)$ is monotonically increasing for $v_i<0$ and zero when $v_i\geq0$.  To this end, we formulate the controller with a stacked-ReLU structure shown in Fig.~\ref{fig:structure_controller}, which is developed in~\cite{cui2020}. This design is a piecewise linear function where the slope of each piece is equal to the summation of weights in activated neurons. Then the requirement 3) can be satisfied by directly thresholding the slope.   The neural network controller is constructed as~\eqref{eq:stacked-relu}
\begin{subequations}\label{eq:stacked-relu}
\begin{align}
     u_i(v_i)&= s_i \text{ReLU}(\mathbf{1}v_i+b_i)+z_i \text{ReLU}(-\mathbf{1}v_i+d_i)\\
     \mbox{where }
     &0<\sum_{j=1}^{l} s_i^{j}< k_i, \quad  \forall l=1,2,\cdots,m
     \label{eq:f+2}\\
      & -k_i<\sum_{j=1}^{l} z_i^{j}< 0, \quad \forall l=1,2,\cdots,m
 \label{eq:f-2}\\
     & b_i^{1}=0, b_i^{l}\leq b_i^{(l-1)},\quad \forall l=2,3,\cdots,m\\
     &d_i^{1}=0, d_i^{l}\leq d_i^{(l-1)},\quad \forall l=2,3,\cdots,m\label{eq:f+3}
\end{align}
\end{subequations}
where $m$ is the number of hidden units and $\mathbf{1}\in \mathbb{R}^m$ is the all $1$'s column vector. Variables $s_i=[\begin{matrix}
 s_i^{1}& s_i^{2} & \cdots&s_i^{m}
\end{matrix}]$ and $z_i=[\begin{matrix} z_i^{1}&z_i^{2}&\cdots&z_i^{m}\end{matrix}]$ are the weight vector of bus $i$; $b_i=[\begin{matrix} b_i^{1}& b_i^{2} & \cdots &b_i^{m}\end{matrix}]^\intercal$ and $d_i=[\begin{matrix} d_i^{1}& d_i^{2} & \cdots & d_i^{m}\end{matrix}]^\intercal$ are the corresponding bias vector. The variables to be trained are weights $\bm{\theta}=\{s,b,z,d\}$ in \eqref{eq:stacked-relu}.
The saturation limits can be satisfied by hard thresholding the output of the neural network. 

\begin{figure}[ht]	
	\centering
	\includegraphics[width=3.5in]{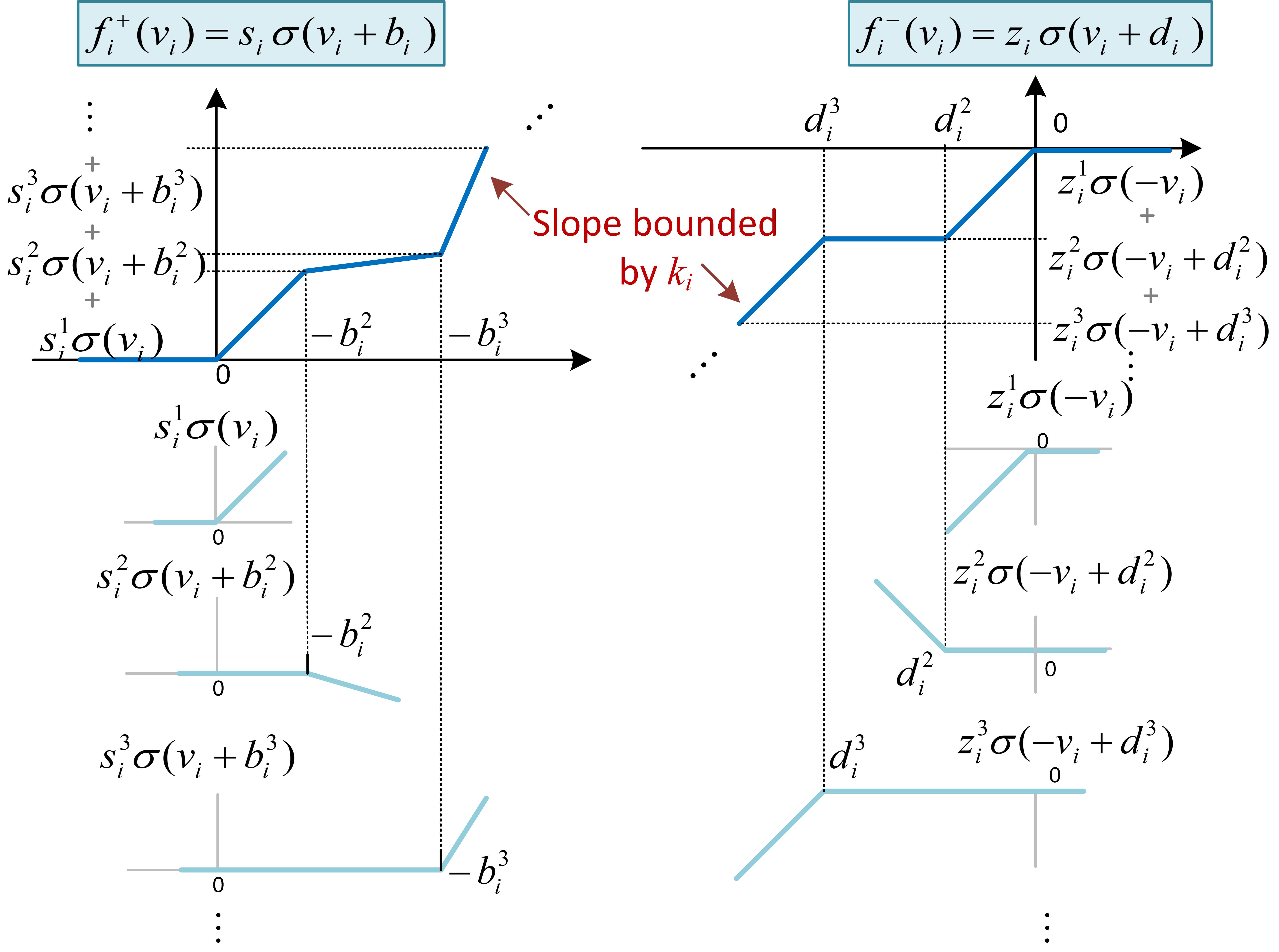}
	\caption{Stacked ReLU neural network to formulate a controller satisfying the stabilizing constraint}
	\label{fig:structure_controller}
\end{figure}

\section{Decentralized Safe Reinforcement Learning} \label{sec:training}

In this section, we construct a decentralized reinforcement learning framework to  optimize neural network controller in each bus locally with observation of trajectories.

Most reinforcement learning algorithms, including Q-learning, actor-critic and DDPG, rely on learning a value function (Q-function) satisfying the Bellman equations. Q-function assumes an infinite-horizon formulation where the states follow a stationary probability distribution, which is generally not true for the voltage control problem in this paper. Instead, REINFORCE policy gradient algorithm adopts the log probability trick and avoids learning the
value function~\cite{sutton2018reinforcement}.
Therefore, we use REINFORCE policy gradient algorithm to obtain sampled gradient for updating the weights of neural network controllers.

Notably, there are natural noises in the system coming from the changes in active power $\bm{p}$, which enable us to implement REINFORCE policy gradient with equivalent stochastic policy. Specifically, we assume that the distribution of noise on the system can be estimated. By incorporating noise term into control action, each action $u_{i, t}$ comes from an equivalent stochastic policy with probability distribution $\pi_{\theta} (u_{i, t}|v_{i, t})$. 
 The gradient for updating weights of neural network controller at bus $i$ is obtained by~\cite{sutton2018reinforcement} 
 \begin{equation}\label{eq:policy_graidient}
    \nabla J(\theta) =  \mathbb{E}[\sum_{t=1}^{T} \nabla_{\theta}\log\pi_{\theta} (u_{i, t}|v_{i, t}) \sum_{t=1}^{T} C_i(u_{i, t})] 
 \end{equation}

The pseudo-code for the decentralized RL framework is given in Algorithm 1. Each bus $i$ has its local RL agent for training in a batch-updating style. 
Let $H$ be the number of batches. At each episode, each agent collects trajectory $\left\{v_{i, 1}^h,u_{i, 1}^h,\cdots,v_{i, T}^h,u_{i, T}^h \right \}$ and the corresponding cost $c_i^h=\sum_{t=1}^{T}C_i(u_{i, t}^h)$
for $h=1,\cdots,H$.  
 Adam algorithm is adopted to update weights of neural network controllers with gradient computed through batch average of~\eqref{eq:policy_graidient}.


 \begin{algorithm}
 \caption{Decentralized Reinforcement Learning algorithm with Policy Gradient}
 \begin{algorithmic}[1]
 \renewcommand{\algorithmicrequire}{\textbf{Require:}}
 \renewcommand{\algorithmicensure}{\textbf{Input:}}
 \REQUIRE Learning rate $\alpha$, batch size $H$, trajectory length T, number of episodes $E$\\
 \ENSURE Initial weights $\theta$ for control network
  \FOR {$episode = 1$ to $E$}
  \FOR {agent $ i = 1$ to $N$}
  \STATE Collect trajectories $\left\{v_{i, 1}^h,u_{i, 1}^h,\cdots,v_{i, T}^h,u_{i, T}^h \right \}$ and the corresponding cost $c_i^h=\sum_{t=1}^{T}C_i(u_{i, t}^h)$ for $h=1,\cdots,H$
  \STATE Compute the gradient $ \nabla J_i(\bm{\theta}_i) =  \frac{1}{H}\sum_{h=1}^{H} \sum_{t=1}^{T} \nabla_{\theta}\log\pi_{\theta} (u_{i, t}^h|v_{i, t}^h) c_i^h$
  \STATE  Update weights in the neural network by passing $J_i(\bm{\theta}_i)$ to Adam optimizer:
  $\bm{\theta}_i \leftarrow \bm{\theta}_i-\alpha \nabla J_i(\bm{\theta}_i) $
  \ENDFOR
  \ENDFOR
 \end{algorithmic} 
 \end{algorithm}

\section{Numerical Results}\label{sec:results}
We verify the performance of the proposed safe RL approach on IEEE 33-bus test feeders~\cite{baran1989network}. We first show that unconstrained neural network controllers learned by RL might lead to an unstable system, while the controllers trained by safe RL approach are guaranteed to stabilize the system.
Then, we show that the proposed decentralized RL framework can learn flexible non-linear controllers for different buses that outperform conventional linear control law.

\subsection{Simulation setup}
 The cost function that each controller collectively optimizes is $C(\bm{u})=\sum_{t=1}^{T}\left( ||\bm{v}_{t}||_1+\gamma||\bm{u}_t||_1\right)$, where $\gamma$ acts as a trade-off parameter and is set to be 0.01.  
The base unit for power and voltage is 100kVA and 12.66kV, respectively.   The  bound  on  action $\Bar{\bm{u}}$  is generated to  be  uniformly  distributed  in $[0.01, 0.05]$. 
We   use   TensorFlow   2.0   framework   to   build   the   reinforcement  learning  environment. 
The episode number, batch size and the number of neurons are 500, 500, 20, respectively. Parameters of neural network controllers are updated using Adam with learning rate initialized to be 0.003 and decayed every 100 steps with a base of 0.6. We compare the performance of neural network controller designed with and without the safe RL approach, as well as conventional linear controller. All of them are trained using the decentralized RL framework.

\subsection{Necessity of the stabilizing requirement }
Intuitively speaking, if a controller achieves a low loss function after training converges, one might hope that it naturally leads to a stabilizing controller since the trajectory does not blow up to a high cost.  Fig~\ref{fig:Dynamic_Behavior_bad} shows the dynamics of voltage deviation under the neural network controllers trained with and without the safe RL approach. The one without safe RL approach is \emph{unstable} and leads to very large state oscillations (Fig.~\ref{fig:Dynamic_Behavior_bad}(b)). In contrast, the controller with safe RL approach shows good performance in Fig.~\ref{fig:Dynamic_Behavior_bad}(a). Therefore, explicitly constraining the controller structure is necessary.
\begin{figure}[ht]
\centering
\subfloat[ safe RL approach ]{\includegraphics[width=1.7in]{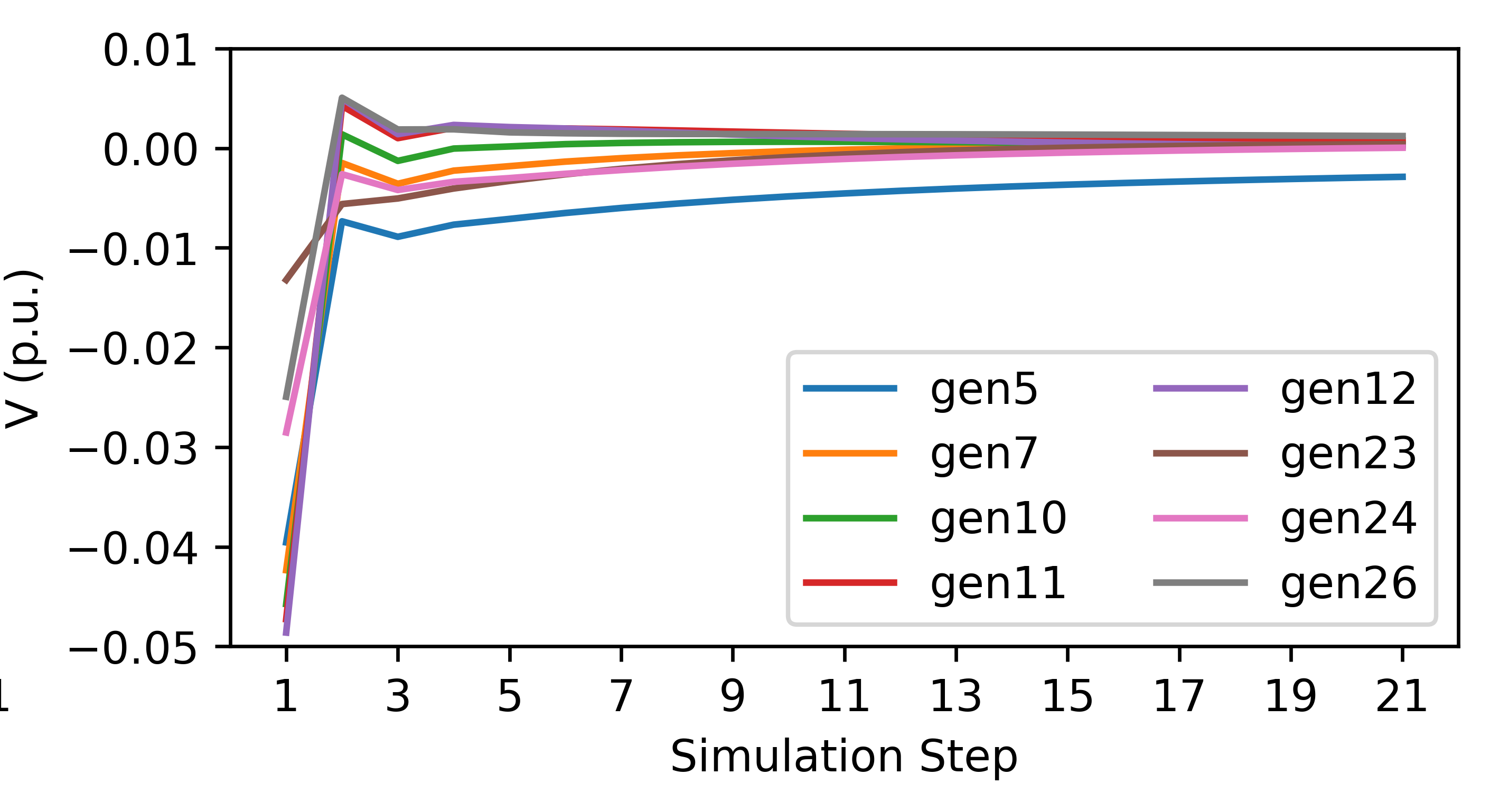}%
}
\subfloat[without safe RL approach]{\includegraphics[width=1.7in]{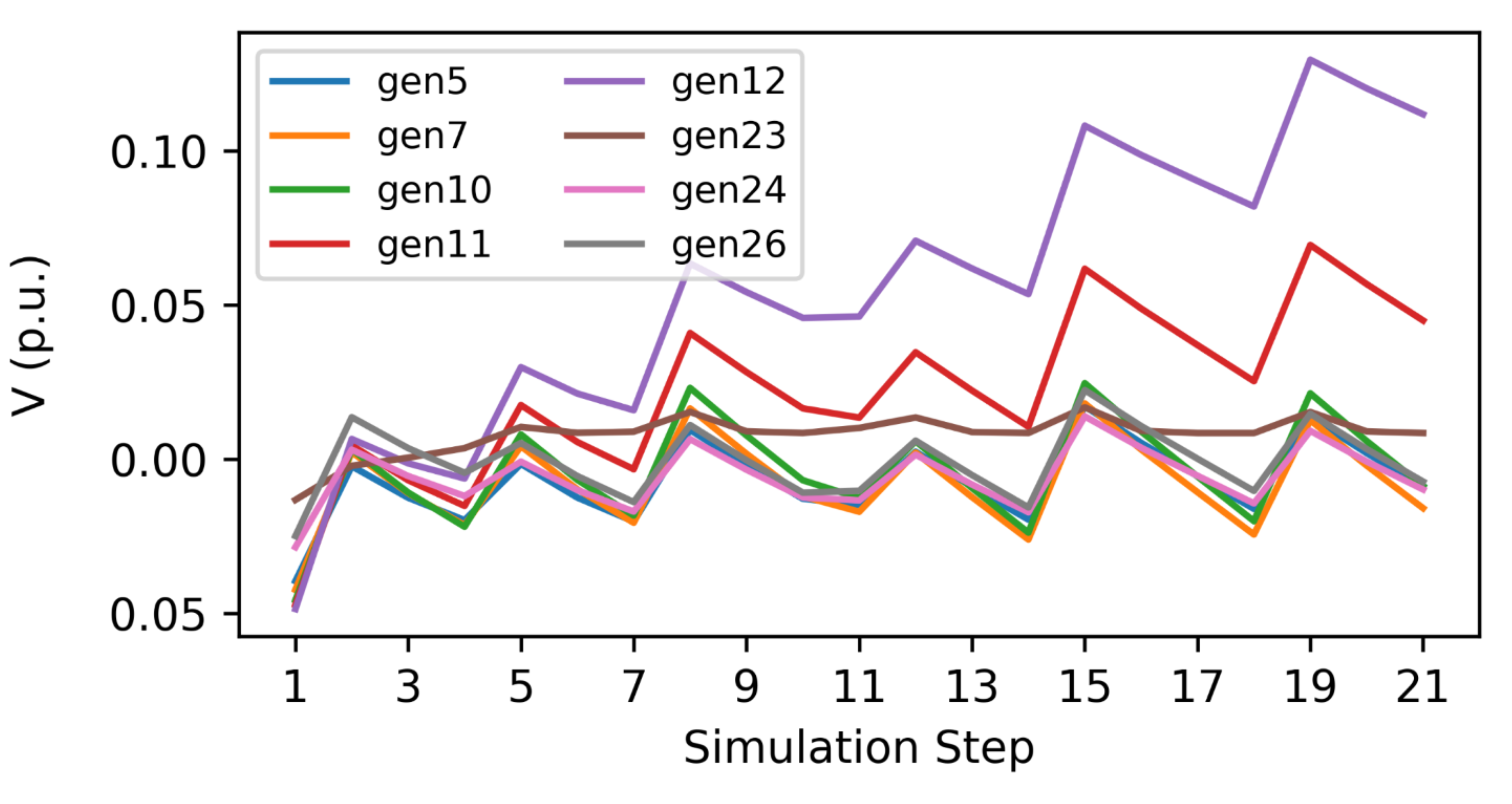}%
}
\caption{Dynamics of voltage deviation for safe RL approach(left) and without safe RL approach(right). The controller designed without the safe RL approach approach leads to unstable trajectories}
\label{fig:Dynamic_Behavior_bad}
\end{figure}

\begin{figure}[ht]	
	\centering
	\includegraphics[width=3.5in]{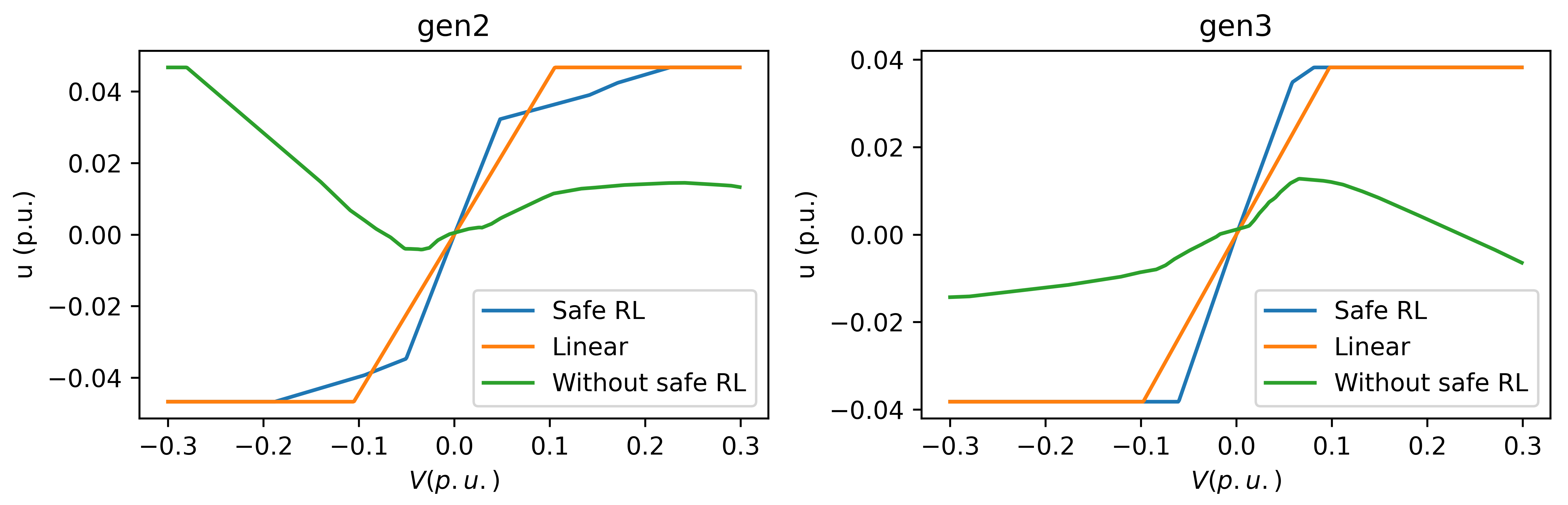}
	\caption{	Voltage control law obtained by linear controller with optimal linear coefficient,  neural network controllers designed with safe RL approach and without  safe RL approach. The neural network controllers learn flexible non-linear control laws for different buses, with the slope of controller obtained by safe RL approach bounded by  Lipschitz  constraints. }
	\label{fig:control_func}
\end{figure}


\subsection{Performance comparison}
To investigate the convergence of the safe RL approach, Fig.~\ref{fig:Cost}(a) shows the normalized cost on the test set along episodes for training of neural network controllers and linear controllers. All the losses converge, with the proposed neural network controllers achieving the lowest cost. Fig.~\ref{fig:Cost}(b) shows the cost on selected buses along the episodes of training.
It is interesting to observe that training the controllers in a decentralized fashion did not impact convergence or performance. Namely, during training, $u_i$ is updated based only on the trajectory of $v_i$, even though the control action impacts the voltage at all neighboring buses.  

\begin{figure}[ht]
\centering
\subfloat[ Total cost ]{\includegraphics[width=1.7in]{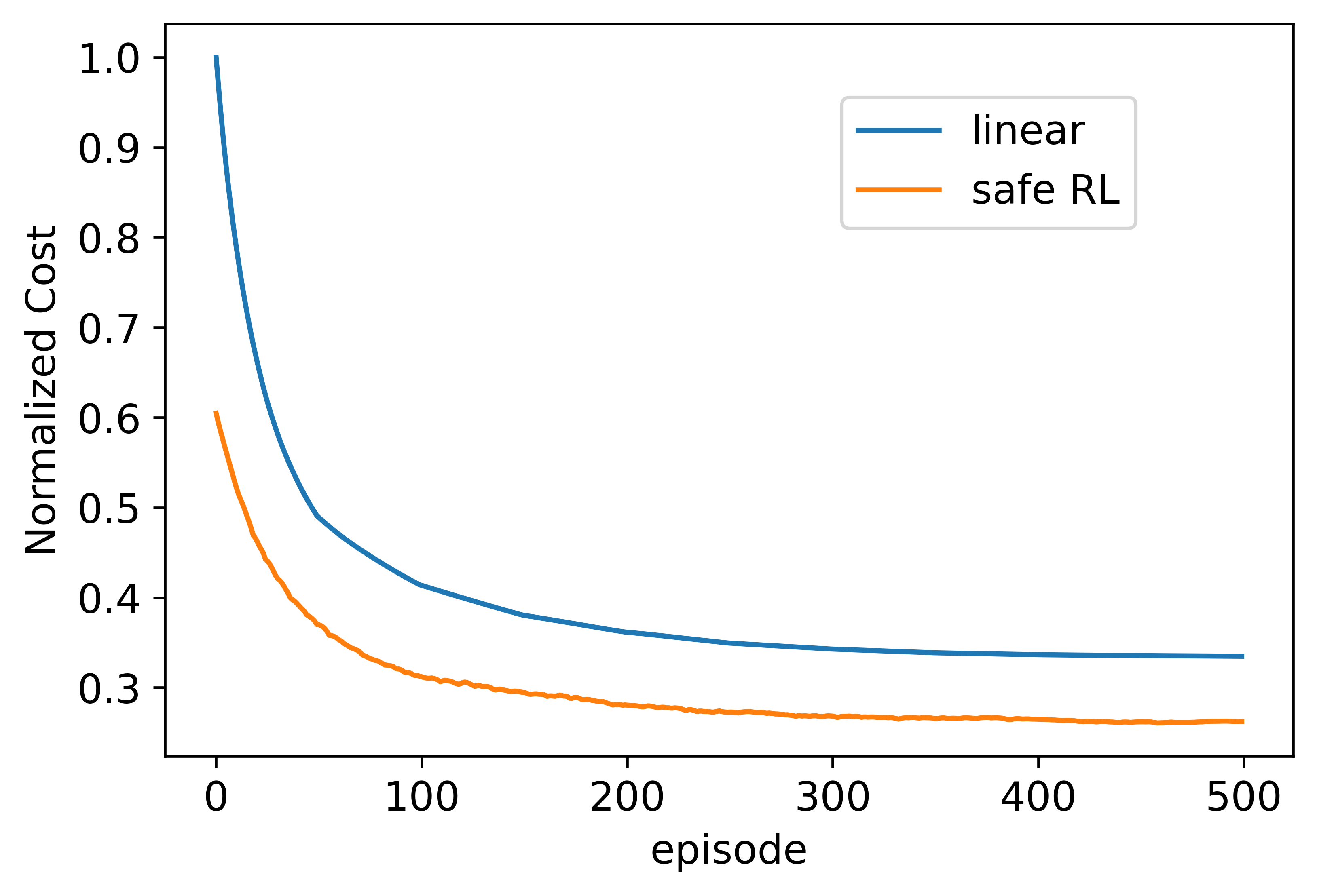}%
}
\subfloat[Cost for selected  buses]{\includegraphics[width=1.7in]{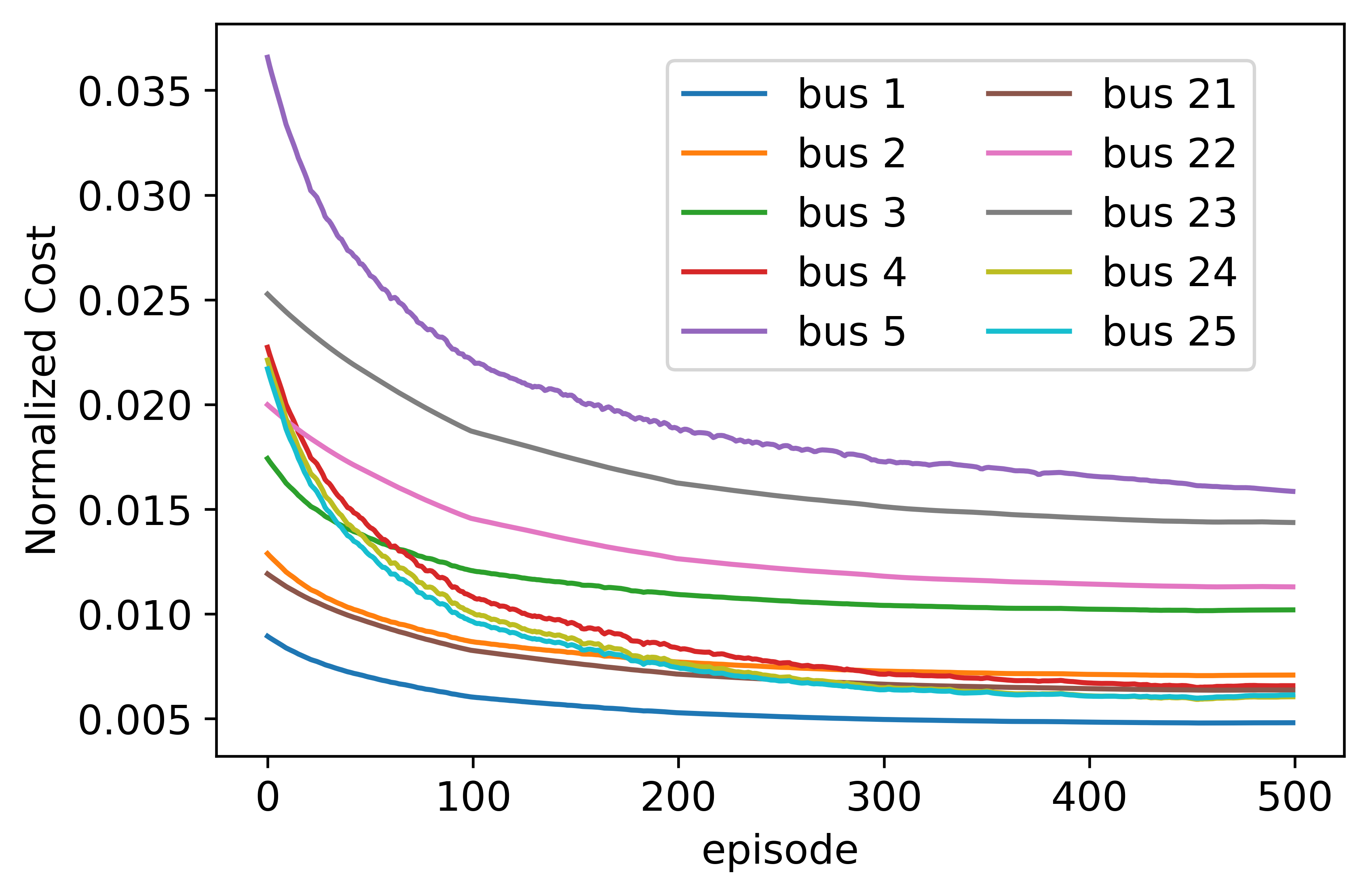}%
}
\caption{Normalized cost on test set along the episode of training. (a) Total cost during training of neural network controller and linear controller. Neural network controller designed with safe RL approach achieves lower cost than conventional linear controller. (b) Cost on selected generator buses during the training of neural network controller. All learning trajectories converge well in the  decentralized model-free setting, even though they interact through the underlying distribution network.  }
\label{fig:Cost}
\end{figure}

The control law for neural network controller learned with safe RL, without safe RL approach and linear controller with optimal linear coefficient are shown in Fig.~\ref{fig:control_func}. The neural network controllers learn flexible non-linear control law for different generators, with the safe RL approach guaranteeing a stabilizing controller by bounding the slope with Lipschitz constraints. 
Fig.~\ref{fig:Compare_RL_linear} illustrates the dynamics of voltage deviation $\bm{v}$ and corresponding control action $\bm{u}$ under optimal linear controller and neural network controller trained by safe RL approach. The neural network controller generally leads to faster decay of voltage deviation.



\begin{figure}[ht]
\centering
\subfloat[Dynamics of $v$ (left) and $u$ (right) for neural network controller obtained through safe RL approach ]{\includegraphics[width=3.5in]{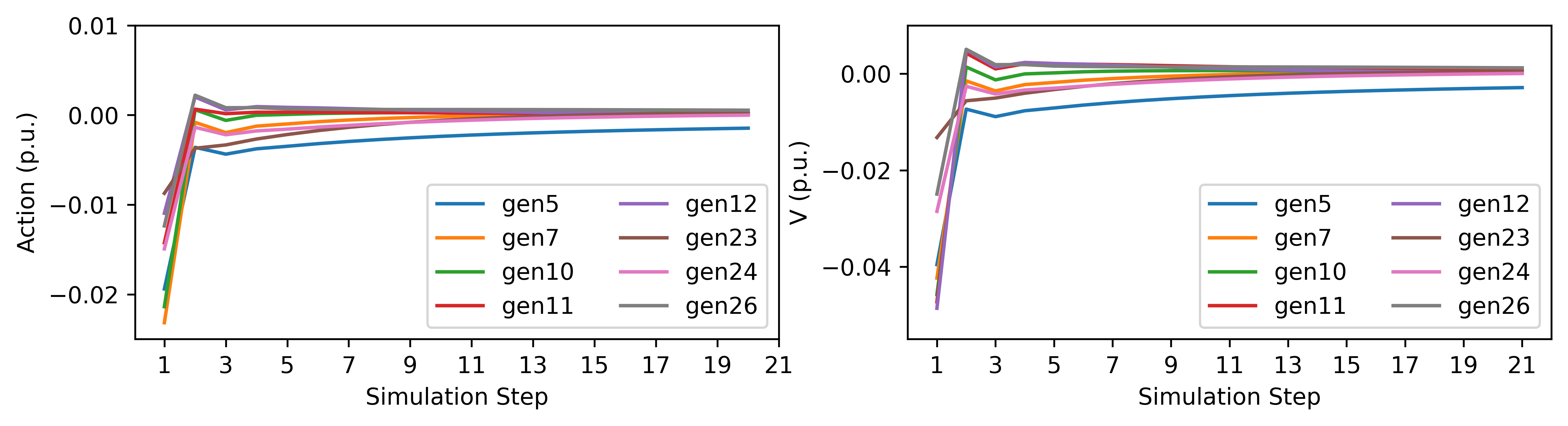}%
}
\hfil
\subfloat[Dynamics of $v$ (left) and $u$ (right) for linear control]{\includegraphics[width=3.5in]{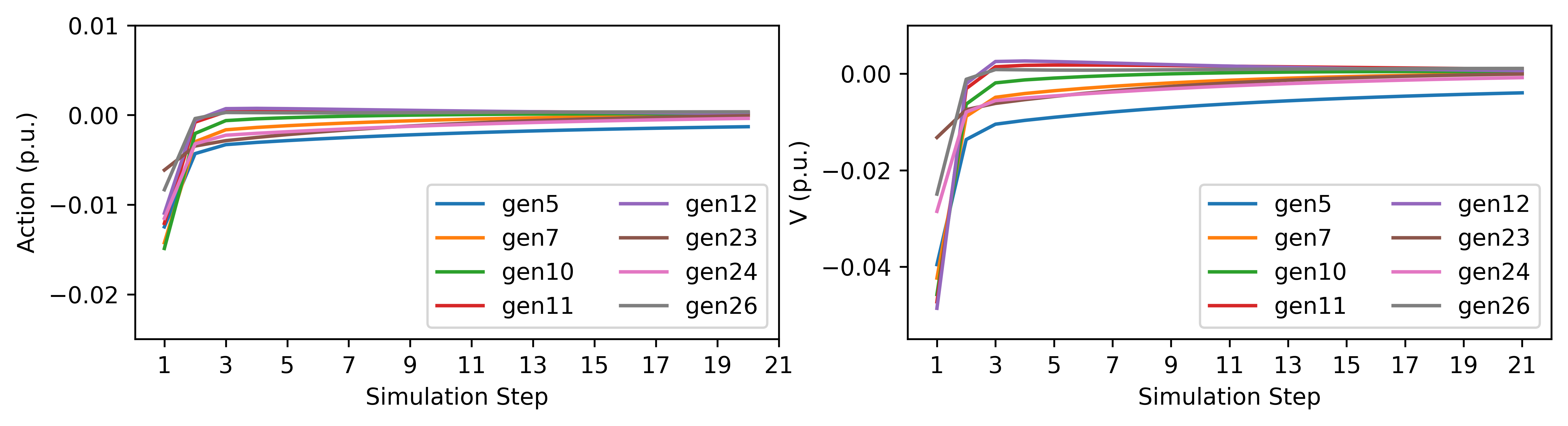}%
}
\caption{Dynamics of  the voltage deviation $\bm{v}$ and the control action $\bm{u}$ in selected generator buses corresponding to (a) neural network controller trained with safe RL approach (b) Linear control obtained by the same decentralized RL algorithm. The neural network controller generally leads to faster decay of voltage deviation. }
\label{fig:Compare_RL_linear}
\end{figure} 

In the test set with random initial states, the distribution of cost in selected buses is shown in Fig~\ref{fig:Cost_distribution}. The average costs of the linear controller, the neural network controller bounded by $\frac{2}{\lambda_{max}(\mathbf{X})}$, and the neural network controller with optimal Lipschiz bound obtained in~\eqref{eq:OPT_Bound}
are 0.44, 0.38 and 0.36, respectively. Therefore, the proposed approach can learn a stabilizing controller that reduces the cost by approximately 18.18\% compared to conventional linear control law. Moreover, safe RL with the optimal Lipschiz bound also reduces the cost by approximately  5.26\% compared to safe RL with the uniform Lipschiz  bound $\frac{2}{\lambda_{max}(\mathbf{X})}$.


\begin{figure}[ht]	
	\centering
	\includegraphics[width=3.5in]{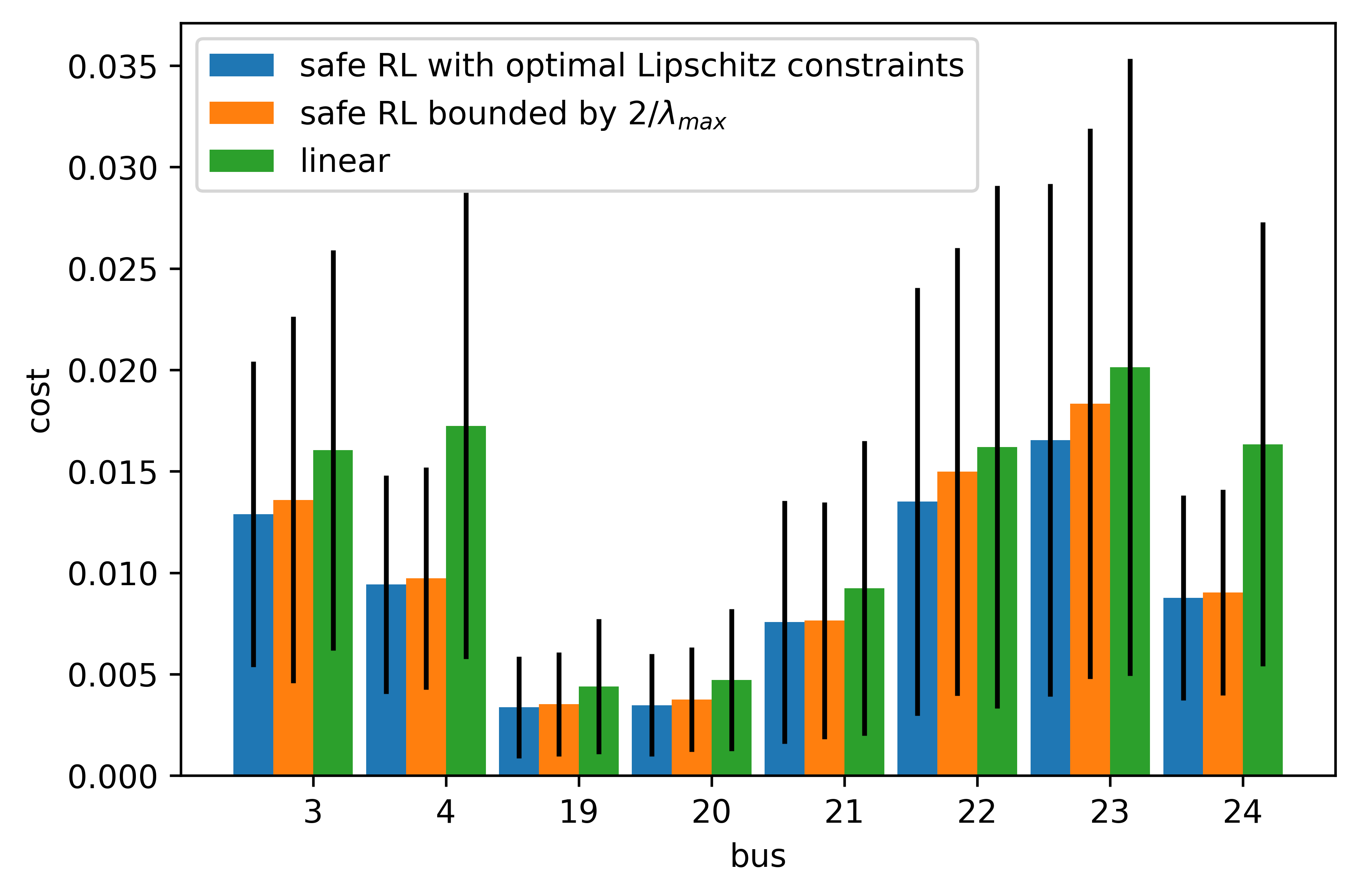}
	\caption{Distribution of cost in selected generator buses  with  random  initial  states corresponding to safe RL with proposed optimal Lipschiz constraints,  safe RL bounded by $\frac{2}{\lambda_{max}(\mathbf{X})}$ and optimal linear control. Compared to uniform bound $\frac{2}{\lambda_{max}(\mathbf{X})}$ and linear controller, the proposed approach reduces the average cost by approximately 5.26\%, 18.18\%, respectively. }
	\label{fig:Cost_distribution}
\end{figure}

\section{Conclusions}\label{sec:conclusions}
This paper proposes a safe RL approach for optimal voltage control. The exponential stability of the system is guaranteed by controllers constrainted by Lipschitz bounds, which are optimized to enlarge the search space. 
The neural network controllers are parameterized by a staked ReLU neural network to satisfy stabilizing constraints implicitly. Each bus updates weights locally with the decentralized RL framework.  Case studies show that  RL without stability constraints  can lead to unstable controllers, while the proposed safe learning approach will lead to a stabilizing controller. The neural network controllers
outperform conventional linear controllers by speeding up the convergence of voltages to reference values with relatively low control effort. Rigorously analysing and comparing decentralized and centralized training is an important future direction for us.

\bibliographystyle{IEEEtran}
\bibliography{Reference.bib}

\end{document}